%% file: main.tex
\documentclass[11pt,a4paper]{article}

\usepackage{hyperref}

\newcommand{\argmax}{\mathop{\rm argmax}}
\newcommand{\argmin}{\mathop{\rm argmin}}

\newcommand{\MD}{\textsf{Modularity-Density}($|\mathcal{C}|\geq 2$)}
\newcommand{\MC}{\textsf{3-Regular-Max-Cut}}
\newcommand{\AP}{\textsf{Auxiliary-Problem}}
\newcommand{\Clique}{\textsf{$(n-4)$-Regular-$k$-Clique}}

\usepackage{amsmath}
\usepackage{amsfonts}
\usepackage{amsthm}
\usepackage{amssymb}
\usepackage{latexsym}   
\usepackage{color}
\usepackage{hyperref}
\usepackage[linesnumbered,ruled,vlined]{algorithm2e}
\usepackage{booktabs}
\usepackage{fullpage}
\usepackage{subcaption}
\usepackage{bm}
\usepackage{url}
\usepackage{cite}
\usepackage{graphicx}
\usepackage{authblk}

\newtheorem{theorem}     {Theorem}
\newtheorem{lemma}       {Lemma}

\theoremstyle{definition}

\newtheorem{problem}     {Problem}

\title{A study on modularity density maximization:\\ Column generation acceleration and computational complexity analysis}

\author[1]{Issey Sukeda\thanks{sukeda-issei006@g.ecc.u-tokyo.ac.jp}}
\author[1]{Atsushi Miyauchi\thanks{miyauchi@mist.i.u-tokyo.ac.jp (Corresponding author)}}
\author[1,2]{Akiko Takeda\thanks{takeda@mist.i.u-tokyo.ac.jp}}
\affil[1]{{\normalsize Graduate School of Information Science and Technology, University of Tokyo, Japan}}
\affil[2]{{\normalsize Center for Advanced Intelligence Project, RIKEN, Japan}}

\date{\today}

\begin{document}

\maketitle

\begin{abstract}
\input{0_abstract}
\end{abstract}


\input{1_introduction}

\input{2_related}
\input{3_method}

\input{4_proposed}

\input{5_computational_results}
\input{6_nphardness}

\input{7_conclusion}

\section*{Acknowledgments}
The authors would like to thank Rafael de Santiago and Lu\'is C. Lamb for sharing the code of their algorithm, 
which makes it easier to implement their algorithm in our experimental environment. 
The authors would also like to thank Keisuke Sato and Yoichi Izunaga for sharing the code of their algorithm, 
which is directly used in our experiments. 
A.M. is supported by Grant-in-Aid for Early-Career Scientists (No.~19K20218). 
A.T. is supported by Grant-in-Aid for Scientific Research (B) (No.~19H04069).

\bibliographystyle{abbrv}
\bibliography{main}

\end{document}

%% file: 0_abstract.tex
Community detection is a fundamental network-analysis primitive with a variety of applications in diverse domains. 
Although the modularity introduced by Newman and Girvan (2004) has widely been used as a quality function for community detection, it has some drawbacks. 
The modularity density introduced by Li et al. (2008) is known to be an effective alternative to the modularity, which mitigates one of the drawbacks called the resolution limit. 
A large body of work has been devoted to designing exact and heuristic methods for modularity density maximization, without any computational complexity analysis.  
In this study, we investigate modularity density maximization from both algorithmic and computational complexity aspects. 
Specifically, we first accelerate column generation for the modularity density maximization problem. 
To this end, we point out that the auxiliary problem appearing in column generation can be viewed as a dense subgraph discovery problem. 
Then we employ a well-known strategy for dense subgraph discovery, called the greedy peeling, for approximately solving the auxiliary problem. 
Moreover, we reformulate the auxiliary problem to a sequence of $0$--$1$ linear programming problems, 
enabling us to compute its optimal value more efficiently and to get more diverse columns. 
Computational experiments using a variety of real-world networks demonstrate the effectiveness of our proposed algorithm. 
Finally, we show the NP-hardness of a slight variant of the modularity density maximization problem, 
where the output partition has to have two or more clusters, as well as showing the NP-hardness of the auxiliary problem.


%% file: 1_introduction.tex
\section{Introduction}
Community detection, identifying community structure in networks, 
is a fundamental network-analysis primitive with a variety of applications in diverse domains~\cite{Fortunato10,Fortunato+16}. 
Roughly speaking, a community in a network is a subset of vertices densely connected internally but only sparsely connected to the rest of the network, 
which corresponds to a social group in friendship networks (trivially), 
a molecular complex in protein--protein interaction networks~\cite{Spirin+03}, 
a spam link farm in the Web graph~\cite{Gibson+05}, and so forth. 

Newman and Girvan~\cite{Newman+04} introduced a quality function for community detection, 
called the modularity. 
Let $G=(V,E)$ be an undirected graph consisting of $n=|V|$ vertices and $m=|E|$ edges. 
For a partition $\mathcal{C}$ of $V$ (i.e., $\mathcal{C}=\{C_1,\dots, C_k\}$ such that $\bigcup_{i=1}^k C_i=V$ and $C_i\cap C_j=\emptyset$ for $i\neq j$), 
the modularity is defined as 
\begin{align*}
Q(\mathcal{C})=\sum_{C\in \mathcal{C}}\left(\frac{|E(C)|}{m}-\left(\frac{\sum_{v\in C}\deg(v)}{2m}\right)^2\right), 
\end{align*}
where $E(C)$ is the set of edges in the subgraph induced by $C\subseteq V$, i.e., 
$E(C)=\{\{i,j\}\in E\mid i,j\in C\}$, and $\deg(v)$ is the degree of $v\in V$ in $G$. 
Intuitively speaking, the modularity quantifies the actual fraction of the number of edges within the clusters minus the expected value of such a fraction 
assuming that the edges are put at random under the same degree distribution. 

The modularity has become the most popular quality function for community detection; 
that is, community detection is often conducted through modularity maximization~\cite{Fortunato10}. 
Specifically, the modularity maximization problem asks to find a partition $\mathcal{C}$ of $V$ 
that maximizes the modularity $Q(\mathcal{C})$. 
Unlike most traditional clustering problems (e.g., the $k$-means problem), the modularity maximization problem 
does not specify the number of clusters in advance; it can be determined by algorithms endogenously. 
The modularity maximization problem is known to be NP-hard~\cite{Brandes+07}. 
To date, a large number of exact, approximation, and heuristic methods have been developed for the problem (see Section~\ref{sec:related}). 

However, it is also known that the modularity has some drawbacks as a quality function for community detection, 
such as the degeneracy~\cite{Good+10} and the resolution limit~\cite{Fortunato+07}. 
The former means that there are numerous solutions that have modularity values very close to the optimal value, 
causing the significant difficulty of computing an optimal solution (even in exponential time). 
The latter means that the modularity fails to detect a community even if it is a relatively large clique connected by only a few edges with the rest of the network, 
when the size of the whole network becomes large. 

To overcome this issue, Li et al.~\cite{Li+08} introduced a novel quality function, which they refer to as the modularity density. 
For a partition $\mathcal{C}$ of $V$, the modularity density is defined as 
\begin{align*}
D(\mathcal{C})=\sum_{C\in \mathcal{C}}\frac{2|E(C)|-|E(C,V\setminus C)|}{|C|}, 
\end{align*}
where $E(C,V\setminus C)$ is the set of edges connecting a vertex in $C$ and a vertex in $V\setminus C$. 
The modularity density quantifies the following value for each cluster and takes the summation: 
the average degree to the inside of the cluster minus the average degree to the outside of the cluster. 
The modularity density is known to be an effective alternative to the modularity; 
in fact, it mitigates one of the aforementioned drawbacks, the resolution limit. 
The modularity density maximization problem asks to find a partition $\mathcal{C}$ of $V$ that maximizes $D(\mathcal{C})$, 
which again does not specify the number of clusters in advance. 

A large body of work has been devoted to designing exact and heuristic methods 
for the modularity density maximization problem. 
Exact methods have been actively studied from mathematical programming point of view~\cite{Costa15,Costa+17,Sato+19}. 
Costa~\cite{Costa15} and Costa et al.~\cite{Costa+17} designed exact methods using mixed-integer-linear-programming reformulation techniques.
Sato and Izunaga~\cite{Sato+19} developed a branch-and-price method, i.e., a branch-and-bound method that incorporates a column generation algorithm for solving a linear programming relaxation on each node in the search tree. 
On the other hand, heuristic methods have been designed based on various techniques such as local search~\cite{Santiago+17_2}, spectral optimization~\cite{Botta+16}, 
simulated annealing~\cite{Zhang+10}, mathematical programming~\cite{Costa+16,deSantiago+17,Izunaga+20}, post processing~\cite{Shang+17}, and others~\cite{Santiago+17}. 
In particular, the algorithm designed by de Santiago and Lamb~\cite{deSantiago+17} is the first column generation algorithm for the modularity density maximization problem, which outputs an optimal solution if the linear programming relaxation used has an integral optimal solution (and it is often the case). 

Unlike the modularity maximization problem, there are no computational complexity results for the modularity density maximization problem; 
in particular, it is not known whether the modularity density maximization problem is NP-hard, as mentioned by Costa et al.~\cite{Costa+17}. 
Some literature (e.g., Li et al.~\cite{Li+08}) states that the problem is NP-hard, 
but there is no valid support for the statement.

\subsection{Our contribution}
In this study, we investigate modularity density maximization from both algorithmic and computational complexity aspects. 

Specifically, we first accelerate column generation for the modularity density maximization problem. 
Column generation is a well-known strategy for solving (linear programming relaxations of) combinatorial optimization problems. 
For the modularity density maximization problem, as mentioned above, de Santiago and Lamb~\cite{deSantiago+17} established a column generation framework
and later Sato and Izunaga~\cite{Sato+19} incorporated it into a branch-and-bound method to design a branch-and-price method. 

To describe our contribution, we review the mechanism of column generation in a general manner. 
First the targeted combinatorial optimization problem is formulated as an integer linear programming problem with exponential number of variables, 
and then its linear programming relaxation is obtained by relaxing the integral constraints. 
It is empirically known that such a relaxation tends to have an integral optimal solution 
(e.g., see de Santiago and Lamb~\cite{deSantiago+17} for the case of the modularity density maximization problem); 
therefore, in that case, the goal can usually be achieved by solving the relaxation. 
However, as the relaxation also has an exponential number of (continuous) variables, it is still hard to solve (even store) the relaxation. 
To overcome this issue, the dual problem of the relaxation is then considered, which has an exponential number of constraints instead. 
Column generation first solves the dual problem with only a part of constraints and obtains an optimal solution to the subproblem. 
Then it solves an optimization problem called the auxiliary problem, 
which finds a constraint in the (full) dual problem violated by the incumbent solution to the most extent (if exists). 
If there is no violated constraint found, it can be seen that the incumbent solution is an optimal solution to the (full) dual problem and we are done. 
Otherwise we add the constraint found to the subproblem, solve it to obtain an updated solution, and repeat the above procedure. 
The auxiliary problem is often computationally expensive; 
therefore, column generation algorithms usually solve the auxiliary problem just approximately using some heuristic methods, 
and solve it exactly only if the heuristic methods fail to find a violated constraint. 
In the context of the modularity density maximization problem, 
de Santiago and Lamb~\cite{deSantiago+17} devised a simple local search heuristic for solving the auxiliary problem approximately and introduced 
$0$--$1$ quadratic programming formulations for solving it exactly. 

In this study, to design an effective heuristic method for the auxiliary problem appearing in column generation for the modularity density maximization problem, 
we point out that the auxiliary problem can be seen as a dense subgraph discovery problem. 
Dense subgraph discovery, aiming to find a (single) dense component in a network, is a well-studied task in the data mining community (see Section~\ref{sec:related}). 
An important fact is that there is a well-known strategy for designing effective dense subgraph discovery algorithms, called the greedy peeling (see e.g., \cite{Boob+20,Charikar00,Chekuri+22,Kawase+18,Kortsarz+94,Kuroki+20,Miyauchi+18,Tsourakakis_15,Tsourakakis+13,Veldt+21}). 
In this strategy, a vertex with the minimum contribution (e.g., the minimum degree) in the current remaining graph is removed iteratively 
to produce a sequence of vertex subsets from $V$ to a singleton, and then the best subset among the sequence (in terms of the objective function at hand) is returned. 
For the classic dense subgraph discovery problem, called the densest subgraph problem, this strategy is known to produce a $1/2$-approximation algorithm~\cite{Charikar00,Kortsarz+94}. 
We apply the greedy peeling to the auxiliary problem in a sophisticated way. 

Moreover, to solve the auxiliary problem exactly (when the above heuristic method fails to find a violated constraint), 
we reformulate the auxiliary problem to a sequence of $0$--$1$ linear programming problems, indexed by the solution size (from $1$ to $n$), 
where one of the optimal solutions to the problems (with the highest objective function value) corresponds to an optimal solution to the auxiliary problem. 
This formulation has two advantages over the above existing $0$--$1$ quadratic programming formulations. 
First, due to the linearity of the objective functions, it can compute the optimal value to the auxiliary problem more efficiently 
(despite the increase of the number of problems to solve). 
Second, it computes not only an optimal solution to the auxiliary problem but also quasi-optimal solutions with various sizes (from $1$ to $n$), 
which often correspond to violated constraints in the (full) dual problem. 

Computational experiments using a variety of real-world networks demonstrate the effectiveness of our proposed column generation algorithm. 
In particular, our algorithm obtains an optimal solution within only two hours
to an instance with more than 100 vertices that cannot be solved by the state-of-the-art methods in 24 hours.

Finally, we show the NP-hardness of a slight variant of the modularity density maximization problem, 
where the output partition has to have two or more clusters. 
Therefore, the variant rules out the trivial clustering consisting of the single cluster containing all vertices. 
Although the NP-hardness of the variant does not lead to the NP-hardness of the original problem (i.e., an open problem mentioned by e.g, \cite{Costa+17,Costa+16,Chen+18,Costa14}), 
ours is the first result to verify the hardness of maximizing the modularity density function. 
In our proof, we construct a polynomial-time reduction from the well-known maximum cut problem (on $3$-regular graphs) 
to the above variant of the modularity density maximization problem (on regular graphs). 
Our reduction is based on an existing reduction from the maximum cut problem to the uniform sparsest cut  problem~\cite{Bonsma+12}; 
however, the difference from the existing reduction is not totally trivial. 
As a by-product of our reduction, we also show that the uniform sparsest cut problem is NP-hard even on regular graphs. 
In addition, we also prove the NP-hardness of the above auxiliary problem appearing in column generation, 
which justifies the use of effective heuristic methods for it.

\subsection{Paper organization}
The rest of the paper is organized as follows: 
In Section~\ref{sec:related}, we review some related work. 
In Section~\ref{sec:method}, we describe a general column generation framework for the modularity density maximization problem. 
In Section~\ref{sec:proposed}, we present our techniques for accelerating column generation. 
We conduct the computational experiments and evaluate the performance of our proposed algorithm in Section~\ref{sec:experiments}. 
We show the NP-hardness of the variant of the modularity density maximization problem and the auxiliary problem in Section~\ref{sec:complexity}. 
We conclude the study and suggest future directions in Section~\ref{sec:conclusion}.

%% file: 2_related.tex
\section{Related Work}\label{sec:related}

The modularity density is known to be an effective alternative of the original modularity. 
de Santiago and Lamb~\cite{deSantiago+20} have recently conducted a systematic comparison between the modularity and the modularity density, 
using exact solutions for both quality functions. 
The modularity density itself also has many variations. 
For example, Chen et al.~\cite{Chen+18} pointed out that the modularity density still has the (milder) resolution limit, 
and developed a variant for further mitigating it. 
Chen et al.~\cite{Chen+14} presented a variant suitable for overlapping community detection, 
where each vertex may belong to more than one cluster. 
Yan et al.~\cite{Yan+20} introduced a variant applicable to community detection in signed networks~\cite{Yan+20}, 
where each edge has a positive/negative label. 


For the modularity maximization problem, Aloise et al.~\cite{Aloise+10} developed a column generation algorithm. 
They used an effective heuristic method for the targeted auxiliary problem, based on a variable neighborhood search, 
and succeeded in solving an instance with more than 500 vertices in reasonable time. 
The design of exact methods for the modularity maximization problem is based on a simple integer linear programming formulation~\cite{Agarwal+08,Dinh+15_exact}. 
This can be done due to the fact that the modularity can be rewritten as a linear form using 0--1 variables over the pairs of vertices. 
Such a property cannot be observed for the modularity density, which makes it difficult to design such a simple exact algorithm for the modularity density maximization problem. 

There are several approximation algorithms for the modularity maximization problem. 
Dinh et al.~\cite{Dinh+15} developed a semidefinite-programming-based approximation algorithm,
and later Kawase et al.~\cite{Kawase+21} improved the algorithm and analysis. 
Heuristic methods for the modularity maximization problem are based on various techniques 
such as greedy procedure~\cite{Blondel+08,Clauset+04}, spectral optimization~\cite{Newman06,Newman06_2,Richardson+09}, simulated annealing~\cite{Guimera+05,Massen+05}, 
and mathematical programming~\cite{Agarwal+08,Cafieri+14,Cafieri+11,Miyauchi+13}, 
some of which have inspired the design of some heuristics for the modularity density maximization problem. 
The most popular heuristic method is the so-called Louvain method~\cite{Blondel+08}, which works well in practice in terms of both solution quality and computation time. 

Dense subgraph discovery is a well-known graph-mining task with a variety of applications~\cite{Gionis+15}. 
The most well-studied optimization problem for dense subgraph discovery would be the densest subgraph problem. 
For a vertex subset $S\subseteq V$, the degree density (or simply density) is defined as $d(S)=\frac{|E(S)|}{|S|}$, 
which is equal to half the average degree of the subgraph induced by $S$. 
The densest subgraph problem asks to find $S\subseteq V$ that maximizes $d(S)$ without any size constraint. 
Unlike most traditional optimization problem for dense subgraph discovery such as the maximum clique problem, 
the densest subgraph problem can be solved exactly in polynomial time; 
indeed, there are exact algorithms based on a maximum-flow-based technique~\cite{Goldberg84} or a linear-programming-based technique~\cite{Charikar00}. 
As mentioned in the introduction, the greedy peeling strategy produces a $1/2$-approximation algorithm for the problem~\cite{Charikar00,Kortsarz+94}, 
which can be implemented to run in linear time, using some sophisticated data structure. 

The densest subgraph problem has many other problem variations. 
The most well-studied variants are size-restricted ones, e.g., the densest $k$-subgraph problem, 
where given $G=(V,E)$ and positive integer $k$, we are asked to find $S\subseteq V$ that maximizes $d(S)$ (or simply $|E(S)|$) 
under the size constraint $|S|=k$. 
It is known that such a restriction makes the problem harder to solve; in fact, the densest $k$-subgraph problem is NP-hard (the maximum clique problem is trivially reducible). 
A variant recently introduced by Miyauchi and Kakimura~\cite{Miyauchi+18} is most related to this study. 
Their variant asks to find $S\subseteq V$ that maximizes $f_\alpha(S)=\frac{|E(S)|-\alpha |E(S,V\setminus S)|}{|S|}$, 
where $\alpha \in [0,\infty)$ is a constant. 
Thus, setting $\alpha = 1/2$, we see that the optimization problem becomes the one-term variant of the modularity density maximization problem, 
which has a strong connection to the auxiliary problem appearing in column generation (see Section~\ref{sec:method}). 
For the above problem with general $\alpha \in [0,\infty)$, the authors developed polynomial-time exact and approximation algorithms, 
based on those for the original densest subgraph problem. 
The acceleration techniques developed in our column generation algorithm are inspired by their algorithm. 


%% file: 3_method.tex
\section{Column Generation Framework}\label{sec:method}
In this section, we review a general column generation framework for the modularity density maximization problem. 
The following framework is due to de Santiago and Lamb~\cite{deSantiago+17}, although there is a slight difference in the notations used. 

Let $G=(V,E)$ be an undirected graph consisting of $n=|V|$ vertices and $m=|E|$ edges. 
Let $\mathcal{C}$ be a partition of $V$. 
We define the contribution of $C\in \mathcal{C}$ as 
\begin{align*}
c(C)=\frac{2|E(C)|-|E(C,V\setminus C)|}{|C|}, 
\end{align*}
which is nothing but the term of $C\in \mathcal{C}$ in the modularity density $D(\mathcal{C})$. 
Therefore, we can write $D(\mathcal{C})=\sum_{C\in \mathcal{C}}c(C)$. 
In the column generation framework, the modularity density maximization problem is formulated as a 0--1 linear programming problem with exponential number of variables, 
by viewing the problem as the set partitioning problem. 
Let $\mathcal{S}=2^V\setminus \emptyset$ be the set of all (nonempty) subsets of $V$. 
For $v\in V$ and $S\in \mathcal{S}$, let $a_{vS}$ be a constant equal to 1 if $v\in S$ and 0 otherwise. 
Let $z_S$ be a 0--1 variable that takes 1 if $S\subseteq V$ is used in the solution and 0 otherwise. 
Then the modularity density maximization problem can be formulated as follows: 
\begin{alignat*}{3}
&\text{maximize} &\quad  &\sum_{S\in \mathcal{S}} c(S) z_S\\
&\text{subject to} & &\sum_{S \in \mathcal{S}} a_{vS}z_S = 1 &\quad & (v \in V),\\
&                  & &z_S\in \{0,1\} && (S \in \mathcal{S}). 
\end{alignat*}
The first constraint stipulates that every vertex belongs to exactly one cluster. 
This formulation has an exponential number of 0--1 variables. 
By relaxing the binary constraint, we can obtain the following linear programming relaxation, which we refer to as the master problem (MP): 
\begin{alignat*}{3}
&\text{maximize} &\quad  &\sum_{S\in \mathcal{S}} c(S) z_S\\
&\text{subject to} & &\sum_{S \in \mathcal{S}} a_{vS}z_S = 1,  &\quad & (v \in V),\\
&                  & &z_S\geq 0 && (S \in \mathcal{S}). 
\end{alignat*}
Note that in this relaxation, the constraint $z_S \leq 1$ ($S\in \mathcal{S}$) is omitted 
because it is automatically satisfied due to the first constraint together with the definition of $a_{vS}$. 
In general, the MP does not necessarily admit an integral solution, 
but according to the computational experiments of de Santiago and Lamb \cite{deSantiago+17}, 
the optimal solution tends to have the integrality for most instances. 
Therefore, we aim to solve the MP instead of the original problem. 
However, the MP in turn has an exponential number of (continuous) variables; thus, it is still hard to solve (even store) the problem. 

To overcome this issue, the dual problem (DP) of the MP is then considered: 
\begin{alignat*}{3}
&\text{minimize} &\quad &\sum_{v \in V} \lambda_v\\
&\text{subject to} &    &\sum_{v \in V} a_{vS}\lambda_v \geq c(S) &\quad &(S \in \mathcal{S}),\\
&                  &    &\lambda_v \in \mathbb{R} & &(v \in V). 
\end{alignat*}
Notice that the DP has an exponential number of constraints rather than variables. 
Column generation first solves the DP with only a part of constraints. 
Specifically, taking a tiny subset $\mathcal{S}'\subseteq \mathcal{S}$, we solve the following relaxed dual problem (RDP): 
\begin{alignat*}{3}
&\text{minimize} &\quad &\sum_{v \in V} \lambda_v\\
&\text{subject to} &    &\sum_{v \in V} a_{vS}\lambda_v \geq c(S) &\quad &(S \in \mathcal{S}'\subseteq \mathcal{S}),\\
&                  &    &\lambda_v \in \mathbb{R} & &(v \in V), 
\end{alignat*}
and obtain its optimal solution $\bm{\lambda^*}=(\lambda^*_v)_{v\in V}$. 
Then we solve an optimization problem called the auxiliary problem (AP), which can be formulated as follows: 
\begin{alignat*}{3}
&\text{minimize}&\quad &\sum_{v \in V} a_{vS} \lambda^*_v - c(S)\\
&\text{subject to}&    &S \in \mathcal{S}. 
\end{alignat*}
We see that the AP is an optimization problem that finds a constraint in the DP violated by the incumbent optimal solution $\bm{\lambda^*}$ to the most extent (if exists). 
Indeed, if there exists a solution $S\in \mathcal{S}$ with the objective function value less than zero, then the constraint in the DP corresponding to $S$ is violated by $\bm{\lambda^*}$. 
Otherwise we see that $\bm{\lambda^*}$ satisfies all constraints in the DP and it is an optimal solution to the DP. 
However, the AP is a computationally expensive problem; 
in fact, in Section~\ref{sec:complexity}, we prove that the AP is NP-hard even on $(n-4)$-regular graphs (when regarding $\bm{\lambda^*}$ as a part of the input). 
Therefore, the existing column generation algorithm by de Santiago and Lamb~\cite{deSantiago+17} employs a simple local search heuristic for solving the AP approximately. 
Their algorithm solves the AP exactly using $0$--$1$ quadratic programming formulations only if the heuristic method fails to find a violated constraint.

%% file: 4_proposed.tex
\section{The Proposed Method}\label{sec:proposed}
In this section, we present our proposed techniques for accelerating column generation for the modularity density maximization problem.

\subsection{Greedy peeling algorithm for approximately solving the auxiliary problem}

Recalling that for $v\in V$ and $S\in \mathcal{S}$, $a_{vS}=1$ if and only if $v\in S$ and multiplying the objective function of the AP by $-1$, 
we can rewrite the AP as follows:
\begin{alignat*}{3}
&\text{maximize}&\quad &c(S) - \sum_{v \in S} \lambda^*_v \\
&\text{subject to}&    &S \in \mathcal{S}.
\end{alignat*}

Here we wish to point out that the AP can be seen as a dense subgraph discovery problem. 
The most well-studied optimization problem for dense subgraph discovery would be the densest subgraph problem, 
which asks to find $S\subseteq V$ that maximizes the density $d(S)=\frac{|E(S)|}{|S|}$. 
As mentioned in Section~\ref{sec:related}, Miyauchi and Kakimura~\cite{Miyauchi+18} introduced a generalization of the density: 
$f_\alpha(S)=\frac{|E(S)|-\alpha |E(S,V\setminus S)|}{|S|}$, where $\alpha \in [0,\infty)$ is a constant, and studied the optimization problem that asks to maximize $f_{\alpha}(S)$. 
As the first term of the objective function of the AP, i.e., $c(S)$, is nothing but $2\cdot f_{1/2}(S)$, 
the AP can be seen as a variant of the optimization problem studied by Miyauchi and Kakimura~\cite{Miyauchi+18}, 
where each vertex $v\in V$ has a weight $\lambda^*_v\in \mathbb{R}$ that is subtracted if $v$ is chosen as an element of $S$. 

The greedy peeling is a well-known strategy for designing effective algorithms for dense subgraph discovery. 
In this strategy, a vertex with the minimum contribution (e.g., the minimum degree) in the current remaining graph is removed iteratively 
to produce a sequence of vertex subsets from $V$ to a singleton, 
and then the best subset among the sequence (in terms of the objective function at hand) is returned. 
For $S\subseteq V$ and $v\in S$, let $\deg_S(v)$ be the degree of $v$ in the induced subgraph $G[S]=(S,E(S))$. 
Then for the densest subgraph problem, the algorithm based on the greedy peeling strategy can be written as Algorithm~\ref{alg:peeling_DSP}. 
\begin{algorithm}[t]
    \caption{Greedy peeling algorithm for the densest subgraph problem}\label{alg:peeling_DSP}
    \SetKwInOut{Input}{Input} 
    \SetKwInOut{Output}{Output} 
	\Input{\ $G = (V,E)$}
	\Output{\ $S\subseteq V$}
    $S_n\leftarrow V$, $i \leftarrow n$\; 
    \While{$i > 1$}{
    	$v_\text{min} \leftarrow \argmin\{\deg_{S_i}(v)\mid v\in S_i\}$\; 
    	$S_{i-1} \leftarrow S_i \setminus \{v_\text{min}\}$, $i\leftarrow i-1$\;
	}
    $S_{\mathrm{max}} \leftarrow  \argmax\{d(S)\mid S\in \{S_1,\dots, S_n\}\}$\;
    \Return $S_{\mathrm{max}}$\;
\end{algorithm}
It is known that Algorithm~\ref{alg:peeling_DSP} is a $1/2$-approximation algorithm for the densest subgraph problem~\cite{Charikar00,Kortsarz+94}. 
Miyauchi and Kakimura~\cite{Miyauchi+18} employed the greedy peeling for their generalized problem 
and demonstrated that the algorithm performs well in practice.  
As the AP is a variant of their optimization problem, it would be reasonable to apply the greedy peeling to the AP. 

From now on, we describe our greedy peeling algorithm for the AP. 
Let $g$ be the objective function of the AP to be maximized, i.e.,
\begin{align*}
g(S)=c(S)-\sum_{v\in S}\lambda^*_v \quad \text{ for }S\subseteq V.
\end{align*}
When solving the AP approximately, we wish to obtain not only a single solution 
but also diverse solutions corresponding to various constraints violated by the incumbent optimal solution $\bm{\lambda^*}$ to the RDP. 
To this end, we introduce a more general objective function $g_p$ ($p\in [0,1]$) by taking a convex combination of the first and second terms of $g$, i.e., for $S\subseteq V$, 
\begin{align}\label{eq:objective}
g_p(S)&=p\cdot c(S)+(1-p)\left(-\sum_{v\in S}\lambda^*_v\right)
=\frac{p(2|E(S)|-|E(S,V\setminus S)|)-(1-p)|S|\sum_{v\in S}\lambda^*_v}{|S|}. 
\end{align}
If we take a large value of $p$, the function emphasizes the first term $c(S)$, 
whereas if we take a small value of $p$, the function emphasizes the second term $-\sum_{v\in S}\lambda^*_v$. 
It is worth mentioning that considering a general objective function like the above might lead to a better solution even in terms of the original objective function $g$, 
because the greedy peeling algorithm is not an exact algorithm (even for the case of the densest subgraph problem). 

The greedy peeling strategy starts with the whole graph and iteratively removes a vertex with the minimum contribution in the current remaining graph. 
Our greedy peeling algorithm introduces two measures for computing contribution of vertices. 
Let $S\subseteq V$. When removing a vertex from $S$, the change of the denominator of \eqref{eq:objective} does not depend on the choice of a vertex. 
Therefore, it suffices to take into account only the numerator of \eqref{eq:objective}. 
Our first measure comes from the following transformation: 
\begin{align*}
\text{The numerator of } \eqref{eq:objective}=\sum_{v\in S}\left(p\left(\deg_S(v)-\deg_{V\setminus (S\setminus \{v\})}(v)\right)-(1-p)|S|\lambda^*_v\right). 
\end{align*}
Note that $\deg_{V\setminus (S\setminus \{v\})}(v)$ represents the degree of $v\in S$ to the outside of $G[S]$. 
It would be reasonable to regard the term of $v\in S$ as a contribution of $v$, leading to our first measure: 
\begin{align*}
\textsf{cont}^\textsf{sum}_p(v) = p\left(\deg_S(v)-\deg_{V\setminus (S\setminus \{v\})}(v)\right)-(1-p)|S|\lambda^*_v \quad \text{ for }v\in S. 
\end{align*}
Our second measure focuses on the actual difference between the objective function values before and after removing a vertex. 
For any $u\in S$, when we remove $u$ from $S$, the numerator of \eqref{eq:objective} decreases by 
\begin{align*}
&p(2|E(S)|-|E(S,V\setminus S)|)-(1-p)|S|\sum_{v\in S}\lambda^*_v\\
&\quad - \left(p(2|E(S\setminus \{u\})|-|E(S\setminus \{u\},V\setminus (S\setminus \{u\}))|)-(1-p)|S\setminus \{u\}|\sum_{v\in S\setminus \{u\}}\lambda^*_v\right)\\
&=p\left(3\deg_S(u)-\deg_{V\setminus(S\setminus \{u\})}(u)\right)-(1-p)\left((|S|-1)\lambda^*_u+\sum_{v\in S}\lambda^*_v\right). 
\end{align*}
Noticing that the last term $\sum_{v\in S} \lambda^*_v$ is independent of the choice of $u\in S$, we can obtain our second measure: 
\begin{align*}
\textsf{cont}^\textsf{diff}_p(v)= p\left(3\deg_S(v)-\deg_{V\setminus(S\setminus \{v\})}(v)\right)-(1-p)(|S|-1)\lambda^*_v \quad \text{ for }v\in S. 
\end{align*}
To increase the diversity of solutions to be obtained, we introduce a more general measure by again taking a convex combination of the above two measures: 
\begin{align*}
\textsf{cont}_{p,q}(v)=q\cdot \textsf{cont}^\textsf{sum}_p(v) + (1-q)\cdot \textsf{cont}^\textsf{diff}_p(v) \quad \text{ for }v\in S. 
\end{align*}

Our greedy peeling algorithm runs as follows: 
Let $P$ and $Q$ be the sets of candidates for the above parameters $p$ and $q$, respectively. 
For each $p\in P$ and each $q\in Q$, the algorithm starts with the whole graph 
and iteratively removes a vertex with the minimum contribution $\textsf{cont}_{p,q}(v)$ in the current remaining graph to obtain a sequence of vertex subsets from $V$ to a singleton. 
Unlike the original greedy peeling strategy, which focuses only on the best subset among the sequence, 
we store all subsets $S\subseteq V$ among the sequence that satisfy $g(S)>0$ 
to find as many constraints as possible in the DP violated by $\bm{\lambda^*}$. 
For reference, the procedure is described in Algorithm~\ref{alg:peeling_ours}. 

\begin{algorithm}[tb]
    \caption{$\textsf{Peeling}(G,\bm{\lambda^*})$}
    \label{alg:peeling_ours}
    \SetKwInOut{Input}{Input} 
    \SetKwInOut{Output}{Output} 
	\Input{\ $G = (V,E)$, $\bm{\lambda^*}=(\lambda^*_v)_{v\in V}$}
	\Output{\ A family $\mathcal{S}$ of vertex subsets}
	$\mathcal{S}\leftarrow \emptyset$\;
	\tcp{The sets $P$ and $Q$ of candidates for the parameters $p$ and $q$, respectively, are fixed in advance.}
	\For{each $p\in P$}{
		\For{each $q\in Q$}{
			$S_n\leftarrow V$, $i \leftarrow n$\; 
    		\While{$i > 1$}{
				\If{$g(S_i)>0$}{$\mathcal{S}\leftarrow \mathcal{S}\cup \{S_i\}$\;}
    			$v_\text{min} \leftarrow \argmin\{\textsf{cont}_{p,q}(v)\mid v\in S_i\}$\; 
    			$S_{i-1} \leftarrow S_i \setminus \{v_\text{min}\}$, $i\leftarrow i-1$\;
			}
		}
	}
    \Return $\mathcal{S}$\;
\end{algorithm}

It is easy to see that Algorithm~\ref{alg:peeling_ours} runs in $O(|P||Q|n^2)$ time, that is, 
each iteration for a pair of $p$ and $q$, corresponding to a usual greedy peeling algorithm, consumes $O(n^2)$ time. 
On the other hand, the most similar algorithm to ours, i.e., the greedy peeling algorithm by Miyauchi and Kakimura~\cite{Miyauchi+18}, 
runs in $O(m+n\log n)$ time. 
This difference is due to the change of the definition of contribution. 
In the algorithm by Miyauchi and Kakimura~\cite{Miyauchi+18}, 
when a vertex is removed, the contribution values are changed only for its neighbors. 
Therefore, using the Fibonacci heap, we can obtain the running time of $O(m+n\log n)$. 
On the other hand, in our algorithm, the contribution values are changed for all vertices, 
due to the term of the current subset size $|S|$, which makes it difficult to reduce the running time of $O(n^2)$.

\subsection{Reformulation of the auxiliary problem}

Here we provide a new reformulation of the AP. 
We deal with the original minimization version of the AP: 
\begin{alignat*}{3}
&\text{minimize}&\quad &\sum_{v \in V} a_{vS} \lambda^*_v - c(S)\\
&\text{subject to}&    &S \in \mathcal{S}.
\end{alignat*}
Note that for any $S\subseteq V$, we have 
\begin{align*}
c(S)=\frac{2|E(S)|-|E(S,V\setminus S)|}{|S|}=\frac{4|E(S)|-\sum_{v\in S}\deg(v)}{|S|}. 
\end{align*}
For each $e\in E$, let $x_e$ be a $0$--$1$ variable that takes $1$ if both of the endpoints of $e$ are included in the solution and 0 otherwise. 
For each $v\in V$, let $y_v$ be a $0$--$1$ variable that takes $1$ if $v$ is included in the solution and $0$ otherwise. 
Then the AP can be formulated as the following $0$--$1$ fractional programming problem: 
\begin{alignat*}{3}
&\text{minimize}  &\ \ &\sum_{v\in V}\lambda_v^* y_v - \frac{\sum_{e \in E} 4x_e - \sum_{v \in V} \deg(v)y_v}{\sum_{v \in V} y_v}\\
&\text{subject to}&    &x_e \leq y_u,\ x_e \leq y_v  &  &(e = \{u,v\} \in E),\\
&                 &    &x_e \in \{0,1\} &&(e \in E), \\
&                 &    &y_v \in \{0,1\} &&(v \in V). 
\end{alignat*}
The fractionality appears in the second term of the objective function, 
corresponding to the contribution of the solution in terms of the modularity density. 
To remove this, we consider fixing the size of solutions $\sum_{v\in V} y_v$ to each $k=1,2,\dots, n$. 
Specifically, for each $k = 1, 2,\dots , n$, we introduce the following $0$--$1$ linear programming problem ($\text{AP}(k)$): 
\begin{alignat*}{3}
&\text{minimize}  &\ \ &\sum_{v\in V} \lambda_v^*y_v - \frac{\sum_{e \in E} 4x_e - \sum_{v\in V} \deg(v)y_v}{k}\\
&\text{subject to}&    &\sum_{v\in V}y_v = k, \\
&                 &    &x_e \leq y_u,\ x_e \leq y_v  &  &(e = \{u,v\} \in E),\\
&                 &    &x_e \in \{0,1\} &&(e \in E), \\
&                 &    &y_v \in \{0,1\} &&(v \in V). 
\end{alignat*}
An optimal solution to $\text{AP}(k)$ is an optimal solution to the AP with the size constraint $|S|=k$. 
Therefore, the best solution among the optimal solutions to $\text{AP}(k)$ ($k=1,2,\dots, n$) is optimal to the AP, 
and its objective value is equal to the optimal value of the AP. 
Owing to this reformulation, the AP becomes a sequence of $0$--$1$ linear programming problems. 

Our reformulation has the following advantages: 
It is expected that an optimal solution to the AP can be found more quickly 
than the $0$--$1$ quadratic programming formulations by de Santiago and Lamb~\cite{deSantiago+17}. 
In fact, $\text{AP}(k)$ is no longer nonlinear and hence easy to handle using MIP solvers, 
despite of the increase of the number of problems from one to $n$. 
Moreover, the reformulation is useful also in terms of finding diverse constraints in the DP violated by $\bm{\lambda^*}$. 
Indeed, in the process of solving the reformulation, we can obtain optimal solutions to $\text{AP}(k)$ ($k=1,2,\dots, n$), 
which can be seen as quasi-optimal solutions to the AP, perhaps corresponding to some violated constraints in the DP. 

Our algorithm for solving the AP exactly using the reformulation is described in Algorithm~\ref{alg:exact_ours}. 
In general, MIP solvers employ a branch-and-bound method to find an optimal solution, 
where many feasible solutions (with good objective function values) are found. 
To obtain diverse constraints, we collect all vertex subsets $S\subseteq V$ 
that appeared as feasible solutions for $\text{AP}(k)$ ($k=1,2,\dots, n$) with negative objective function values. 

\begin{algorithm}[tb]
    \caption{$\textsf{Exact}(G,\bm{\lambda^*})$}
    \label{alg:exact_ours}
    \SetKwInOut{Input}{Input} 
    \SetKwInOut{Output}{Output} 
	\Input{\ $G = (V,E)$, $\bm{\lambda^*}=(\lambda^*_v)_{v\in V}$}
	\Output{\ A family $\mathcal{S}$ of vertex subsets}
	$\mathcal{S}\leftarrow \emptyset$\;
	\For{each $k=1,2,\dots, n$}{
		Solve $\text{AP}(k)$ and let $\mathcal{S}_k$ be the family of vertex subsets that appear as feasible solutions for $\text{AP}(k)$ with negative objective function values\;
		$\mathcal{S}\leftarrow \mathcal{S}\cup \mathcal{S}_k$\;
	}
	\Return $\mathcal{S}$\;
\end{algorithm}

\subsection{Summary: Column generation algorithm using our proposed techniques}

We are now ready to present our column generation algorithm for the modularity density maximization problem. 
To make the RDP have an optimal solution (i.e., bounded), we have to set an initial $\mathcal{S}'$ in the RDP appropriately. 
One trivial way is to use the singletons, i.e., $\mathcal{S}'=\{\{v\}\mid v\in V\}$. 
Note that the output of any algorithm for the modularity density maximization problem can be used as an initial $\mathcal{S}'$. 
Our algorithm first solves the RDP with $\mathcal{S}'$ and obtains its optimal solution $\bm{\lambda^*}$. 
Then it performs the greedy peeling procedure (Algorithm~\ref{alg:peeling_ours}). 
If new violated constraints are found, they are added to the RDP and the process is repeated from the beginning; 
otherwise the algorithm conducts the exact procedure (Algorithm~\ref{alg:exact_ours}). 
For reference, the entire procedure is summarized in Algorithm~\ref{alg:CG_ours}. 

\begin{algorithm}[t]
    \caption{Our column generation algorithm}
    \label{alg:CG_ours}
    \SetKwInOut{Input}{Input} 
    \SetKwInOut{Output}{Output} 
	\SetKw{Continue}{continue}
	\Input{\ $G = (V,E)$}
	\Output{\ An optimal solution $\bm{\lambda^*}=(\lambda^*_v)_{v\in V}$ to the DP}
    $\mathcal{S}'\leftarrow \{\{v\}\mid v\in V\}$\; \tcp{Initialization of the set of constraints in the RDP.}
	\While{\texttt{TRUE}}{	
    	Solve the RDP with $\mathcal{S}'$ and obtain its optimal solution $\bm{\lambda^*}=(\lambda^*_v)_{v\in V}$\;
		$\mathcal{S}_\text{peeling} \leftarrow \textsf{Peeling}(G,\bm{\lambda^*})$\;
		\If{$\mathcal{S}_\mathrm{peeling}=\emptyset$}{
			$\mathcal{S}_\text{exact}=\textsf{Exact}(G,\bm{\lambda}^*)$\;
			\If{$\mathcal{S}_\mathrm{exact}=\emptyset$}{
				\Return $\bm{\lambda}^*$\;
			}
			$\mathcal{S'}\leftarrow \mathcal{S'}\cup \mathcal{S}_\text{exact}$\;
			\Continue\;
		}
		$\mathcal{S'}\leftarrow \mathcal{S'}\cup \mathcal{S}_\text{peeling}$\;
	}
\end{algorithm}

%

%% file: 5_computational_results.tex
\section{Experimental Evaluation}\label{sec:experiments}

In this section, we conduct computational experiments using a variety of real-world graphs and evaluate the performance of our proposed method.

\subsection{Experimental setup}

Here we explain our experimental setup.

\subsubsection{Instances}

\begin{table}[tb]
\centering
\caption{Real-world graphs used in our experiments.}\label{tab:instances}
\scalebox{0.9}{
\begin{tabular}{clrrrc}
\toprule
ID& Name & $n$ & $m$ &Best $D(\mathcal{C})$ & Description\\
\midrule
1&\texttt{Strike}&24&38&$8.86111^{*}$~\cite{Costa15} & Social network~\cite{Michael97}\\
2&\texttt{Karate}&34&78&$7.8451^{*}$~\cite{Costa15}& Social network~\cite{Zachary77}\\
3&\texttt{Dolphins}&62&159&$12.1252^{*}$~\cite{deSantiago+17}& Animal social network~\cite{Lusseau+03}\\
4&\texttt{Les Mis\'erables}&77&254&$24.5474^{*}$~\cite{Sato+19}& Co-appearance network~\cite{Lusseau+03}\\
5&\texttt{Polbooks}&105&441&$21.9652^{*}$~\cite{deSantiago+17}& Co-purchased network~\cite{Krebs}\\
6&\texttt{Adjnoun}&112&425&7.651~\cite{Santiago+17}& Word adjacency network~\cite{Newman06_2}\\
7&\texttt{Football}&115&613&44.340~\cite{Santiago+17}& Sports game network~\cite{Girvan+02}\\
8&\texttt{Jazz}&198&2,742&49.716~\cite{Santiago+17} & Social network~\cite{Gleiser+03}\\
\bottomrule
\end{tabular}
}
\end{table}

Table~\ref{tab:instances} summarizes the instances on which our experiments were conducted.
The first and second columns indicate IDs and names, respectively,  of the instances. 
The third and fourth columns present the number of vertices and the number of edges, respectively. 
The fifth column reports the best previously-known modularity density value for the instances, where the value is associated with ``$*$'' if it is known to be the optimal value. 
The last column presents the description of the instances.

\subsubsection{Baseline methods}

We employ powerful existing algorithms designed by de Santiago and Lamb~\cite{deSantiago+17} and Sato and Izunaga~\cite{Sato+19}, as baseline methods.

The algorithms by de Santiago and Lamb~\cite{deSantiago+17} are the first series of column generation algorithms 
for the modularity density maximization problem. 
The design of the algorithms follows the column generation framework reviewed in Section~\ref{sec:method}, 
but there are some additional techniques for accelerating the computation. 
Among a variety of algorithms they developed, the algorithm called CGII+ILS is found to be the best algorithm. 
CGII+ILS solves the AP using the formulation called AP-II. 
AP-II is a 0--1 convex quadratic programming formulation derived by linearizing the nonconvex terms of the previous formulation called AP-I, 
which is a 0--1 nonconvex quadratic programming formulation produced by multiplying the objective function of the AP by the size term $\sum_{v\in V}y_v$. 
It should be noted that the transformation from the AP to AP-I (and AP-II) does change the ordering of solutions in terms of the objective function value; 
however, it does not change the sign of the objective function values of solutions 
and thus AP-I and AP-II can be used for deciding if the incumbent solution $\bm{\lambda^*}$ of the RDP is feasible for the DP. 
The term ILS of CGII+ILS means that the algorithm employs a heuristic method called the iterative local search for approximately solving the AP. 
The method starts with a randomly chosen $S\subseteq V$. 
Then it repeats the following process for appropriate number of times: 
Along with a random ordering of vertices produced, 
it moves a vertex from or to the cluster if the movement improves the objective function value, 
and then perturbs the cluster if such a movement does not happen. 
To construct an initial $\mathcal{S}'$ in the RDP, CGII+ILS uses the method called the hybrid local search, 
which can be seen as a modularity density counterpart of the Louvain method~\cite{Blondel+08} 
followed by an additional local search procedure. 

The algorithm by Sato and Izunaga~\cite{Sato+19} is the state-of-the-art exact algorithm for the modularity density maximization problem. 
The algorithm is based on a branch-and-price method, 
i.e., a branch-and-bound method that incorporates a column generation algorithm for solving a linear programming relaxation for each node in the search tree. 
The authors incorporated two existing techniques into the method, called the set-packing relaxation~\cite{Sato+12} and the multiple-cutting-plane-at-a-time~\cite{Izunaga+17}. 
As an initial $\mathcal{S}'$ in the RDP, the authors just employed the singletons, i.e., $\mathcal{S}'=\{\{v\}\mid v\in V\}$. 
An important fact is that unlike the algorithms by de Santiago and Lamb~\cite{deSantiago+17} and even our proposed algorithm, 
the algorithm by Sato and Izunaga~\cite{Sato+19} is guaranteed to output an optimal solution 
to the modularity density maximization problem (rather than its relaxation) at the termination. 
Therefore, comparing the performance of those two types of algorithms seems unfair. 
However, the algorithm by Sato and Izunaga~\cite{Sato+19} can be a baseline in the sense that our proposed algorithm needs to have greater scalability than their algorithm.

\subsubsection{Implementation and machine}
Our proposed algorithm was implemented in Python 3. 
Throughout the experiments, we set $P=\{0.0,0.1,\dots, 1.0\}$ and $Q=\{0.0, 0.5, 1.0\}$ in the greedy peeling (Algorithm~\ref{alg:peeling_ours}). 
The algorithm CGII+ILS by de Santiago and Lamb~\cite{deSantiago+17} was also implemented in Python 3; 
however, to perform a fair comparison, the algorithm was modified so that an initial $\mathcal{S}'$ of the RDP is set to be the singletons, 
consistent with our proposed algorithm and the algorithm by Sato and Izunaga~\cite{Sato+19}. 
Note that the choice of initial $\mathcal{S}'$ is critical to the performance of the algorithms; 
indeed, if we set it to an optimal solution to the modularity density maximization problem, the algorithms are expected to terminate shortly. 
The purpose of our experiments is to evaluate the performance of the algorithms in terms of exactly or approximately solving the AP, 
irrespective of the choice of initial $\mathcal{S}'$. 
The parameter called ``factor'' in CGII+ILS, which controls the degree of cluster perturbation, 
was set to 0.7, as recommended in de Santiago and Lamb~\cite{deSantiago+17}. 
As for the algorithm by Sato and Izunaga~\cite{Sato+19}, we used the implementation provided by the authors. 

It should be noted that when implementing the above three algorithms, to mitigate the effect of numerical errors,
we have to modify them slightly so that only the constraints that are violated by the incumbent solution $\bm{\lambda^*}$ 
with some tolerance $\epsilon>0$ are added to the RDP (or the counterpart as for the algorithm by Sato and Izunaga~\cite{Sato+19}). 
Without such a modification, due to numerical errors, an optimal solution to the DP might be decided to be infeasible for the DP,
and the algorithm might not terminate.
In our experiments, the error parameter $\epsilon$ was set to $10^{-6}$. 

For all algorithms, we used Gurobi Optimizer 9.5.0 as a mathematical programming solver, with the default parameter settings. 
All the experiments were conducted on a machine equipped with an Apple M1 chip and 16GB RAM.

\subsection{Experiments and results}

\begin{table}[tb]
\centering
\caption{Performance of our algorithm and the state-of-the-art methods.}\label{tab:result_main}
\scalebox{0.79}{
\begin{tabular}{crrrrrrrrr}
\toprule
&\multicolumn{2}{c}{CGII+ILS~\cite{deSantiago+17}} & \multicolumn{2}{c}{Sato and Izunaga~\cite{Sato+19}} 
& \multicolumn{2}{c}{Ours (Algorithm~\ref{alg:CG_ours})} \\
\cmidrule(lr){2-3} \cmidrule(lr){4-5} \cmidrule(lr){6-7}
ID&time(s) & \#constr. &  time(s) &\#constr. & time(s) & \#constr. & Best $D(\mathcal{C})$ updated\\
\midrule
1&5.9$\pm$0.3&479.4$\pm$53.8                  &0.4              &  62      &\textbf{0.3}     & 780                 & ---\\
2&43.7$\pm$1.9&1056.0$\pm$45.5                  &\textbf{0.7}              &  80       &2.4    & 3157                 & ---\\
3&3261.4$\pm$316.5&7567.6$\pm$603.2                  &\textbf{16.4}             &  185      &42.8    & 21999                  & ---\\
4&13294.0$\pm$1991.1&18309.8$\pm$2036.0                  &\textbf{57.9}            & 205       &119.7   &  59862                 & ---\\
5&$>\text{24h}$ (20.7943)&(70160)                  &\textbf{365.4}            & 380       &1886.3  &184659                   & ---\\
6&$>\text{24h}$ (7.5892)&(2374)                  &$>\text{24h}$  &(11440)        &\textbf{5681.1}  &121894                   & $7.8250^{*}$\\
7&$>\text{24h}$ (34.0075)&(51181)                  &70195.1                 & 426       &\textbf{5860.8}  & 359002                  & $44.3879^{*}$\\
8&$>\text{24h}$ (32.7393)&(8068)                             &$>\text{24h}$           &(571)        &$>\text{24h}$ (47.2153)         &(1265072)                   &--- \\
\bottomrule
\end{tabular}
}
\end{table}

First we compare the performance of our algorithm with the aforementioned baseline methods. 
Time limit was set to 86,400s (24 hours). 
As explained above, CGII+ILS has a stochastic behavior: 
For each instance, if the algorithm does not terminate in the time limit at the first run, we do not run the algorithm for the instance anymore. 
If it is not the case, we execute the algorithm five times and report its statistics. 

The results are shown in Table~\ref{tab:result_main}, 
where we report the computation time needed and the number of constraints added to the RDP 
(or the counterpart as for the algorithm by Sato and Izunaga~\cite{Sato+19}). 
For CGII+ILS, we report the average and standard deviation of the computation time and the number of constraints over the five trials. 
When an algorithm exceeds the time limit, we report the lower bound on the optimal value of the DP 
(only for CGII+ILS and our proposed algorithm) and the number of constraints at the time limit (in parentheses). 
As the algorithm by Sato and Izunaga~\cite{Sato+19} solves the dual problems of relaxations of the MP 
(using the set-packing relaxation) and tightens the relaxation gradually,
the lower bound kept in the algorithm is not necessarily a lower bound on the optimal value of the DP; 
therefore, it is omitted in the table. 
For each instance, the best result among the algorithms (in terms of the computation time) is written in bold. 
We do not apply such a rule to the number of constraints 
because the number of constraints is not necessarily a reasonable indicator of the performance of algorithms; 
indeed, it may be a good strategy to add constraints to the RDP actively to terminate quickly.

As can be seen, our proposed algorithm outperforms the baseline methods. 
Indeed, for all instances except for the largest one (\texttt{Jazz}), 
only our algorithm terminates in reasonable time (within two hours), 
and moreover, the optimal solution to the MP derived by the optimal solution to the DP computed by the algorithm has an integrality for each of those instances; 
thus, it turns out that the algorithm obtains an optimal solution to the modularity density maximization problem for those instances. 
In particular, the optimal solutions to \texttt{Adjnoun} (ID~6) and \texttt{Football} (ID~7) are obtained for the first time; 
the best previously-known modularity density values 7.651 and 44.340 are improved to 7.8250 and 44.3879, respectively, with the optimality guarantee. 
Although the algorithm by Sato and Izunaga~\cite{Sato+19} computes an optimal solution faster for small instances, 
it has poorer performance for large instances. 
Indeed, the algorithm does not even terminate for \texttt{Adjnoun} 
and consumes more than 70,000s for \texttt{Football}.  
The performance of CGII+ILS is not comparable to that of our proposed algorithm or the algorithm by Sato and Izunaga~\cite{Sato+19}.

\begin{table}[tb]
\centering
\caption{Performance of the variants of our algorithm.}\label{tab:result_supp}
\scalebox{0.79}{
\begin{tabular}{crrrrrrrrr}
\toprule
  &\multicolumn{2}{c}{Ours w/o greedy peeling} & \multicolumn{2}{c}{Ours w/o reformulation} \\
\cmidrule(lr){2-3} \cmidrule(lr){4-5} 
ID &time(s) & \#constr. &  time(s) &\#constr. \\
\midrule
1&2.5$\pm$0.3&542.4$\pm$44.3                  &\textbf{1.0}  &781                         \\ 
2&13.7$\pm$1.7&1359.0$\pm$51.8                  &\textbf{4.7}  &3087                         \\
3&\textbf{487.3$\pm$58.1}&8660.8$\pm$713.7                  &1052.2  &23605                         \\
4&\textbf{1194.6$\pm$101.4}&20146.6$\pm$1555.7                  &1230.0  &62038                         \\ 
5&\textbf{12491.4$\pm$1445.2}&61678.6$\pm$5240.7                  &$>\text{10h}$ (21.5496)  & (214623)                        \\
6&\textbf{10185.9$\pm$847.4}&46127.6$\pm$2092.3                  &$>\text{10h}$ (7.7835)  &(118258)                         \\
7&$>\text{10h}$ (44.1272)&(42358)                  &\textbf{$>\text{10h}$ (44.3879)}  &(361468)                         \\
8&$>\text{10h}$ (31.3547)  &(4052)                  &\textbf{$>\text{10h}$ (46.4529)}  &(850531)                      \\ 
\bottomrule
\end{tabular}
}
\end{table}

Finally we investigate how much each of the proposed techniques, the greedy peeling algorithm and the reformulation, contributes to the significant performance of our algorithm. 
To this end, we implement and evaluate the following two variants of our proposed algorithm, 
one replacing the greedy peeling algorithm with the above heuristic method called the iterative local search designed by de Santiago and Lamb~\cite{deSantiago+17}
and the other replacing the reformulation with AP-II by de Santiago and Lamb~\cite{deSantiago+17}. 
Time limit was set to 36,000s (10 hours). 
Note that the first variant has a stochastic behavior due to the heuristic method by de Santiago and Lamb~\cite{deSantiago+17}; 
thus, we run the algorithm and report the results in the same way as CGII+ILS. 

The results are shown in Table~\ref{tab:result_supp}. 
As can be seen, both algorithms have poorer performance than our proposed algorithm but have better performance than CGII+ILS (see also Table~\ref{tab:result_main}). 
Therefore, we see that both of the proposed techniques do contribute to the significant performance of our algorithm. 
However, we cannot conclude which of the techniques affects the performance more significantly; 
indeed, only the algorithm using the reformulation technique succeeds in computing an optimal solution to \texttt{Polbooks} (ID~5) and \texttt{Adjnonu} (ID~6), 
but the algorithm using the greedy peeling obtains a better lower bound for the larger instances \texttt{Football} (ID~7) and \texttt{Jazz} (ID~8). 

%% file: 6_nphardness.tex
\section{Computational Complexity Analysis}\label{sec:complexity}
In this section, we show the NP-hardness of a variant of the modularity density maximization problem and the AP appearing in column generation.

\subsection{NP-hardness of a variant of the modularity density maximization problem}
Here we study the following variant of the modularity density maximization problem: 
Given an undirected graph $G=(V,E)$, we are asked to find a partition $\mathcal{C}$ of $V$ 
that maximizes the modularity density $D(\mathcal{C})$ under the constraint that $|\mathcal{C}|\geq 2$. 
Therefore, this variant rules out the trivial clustering $\mathcal{C}=\{V\}$. 
We call the variant the modularity density maximization problem with at least two clusters. 
We prove the following theorem: 
\begin{theorem}\label{thm:main}
The modularity density maximization problem with at least two clusters is NP-hard even on regular graphs.
\end{theorem}

We first introduce the decision version of the problem: 
\begin{problem}[\MD]
Given an undirected graph $G=(V,E)$ and a real number $r\in \mathbb{R}$, 
does there exist a partition $\mathcal{C}$ of $V$ with $|\mathcal{C}|\geq 2$ that satisfies $D(\mathcal{C})\geq r$?
\end{problem}
To prove Theorem~\ref{thm:main}, it suffices to show the NP-completeness of \MD \ on regular graphs. 
For a given partition $\mathcal{C}$ of $V$ with $|\mathcal{C}|\geq 2$ and a real number $r$, we can check if $D(\mathcal{C})\geq r$ in polynomial time; 
thus, \MD \ is in the class NP. 

To show the NP-completeness, 
we construct a polynomial-time reduction from the decision version of the maximum cut problem on $3$-regular graphs, 
which is known to be NP-complete~\cite{Alimonti+97}. 
Let $G=(V,E)$ be an undirected graph. 
A bipartition $\{X,Y\}$ of $V$ (i.e., a partition of $V$ with size two) is called a cut of $G$. 
The value of cut $\{X,Y\}$ is defined as the number of edges between $X$ and $Y$, i.e., $|\{\{u,v\}\in E\mid u\in X,\ v\in Y\}|$. 
We denote by $\textsf{val}(\{X,Y\})$ the value of cut $\{X,Y\}$. 
Then we can introduce the problem: 
\begin{problem}[\MC]
Given a $3$-regular graph $G=(V,E)$ and a positive integer $k\in \mathbb{Z}_{>0}$, 
does there exist a cut $\{X,Y\}$ of $G$ that satisfies $\textsf{val}(\{X,Y\})\geq k$?
\end{problem}

The following lemma guarantees the existence of the desired reduction: 
\begin{lemma}\label{lem:reduction_to_MD}
There exists a polynomial-time reduction from \MC \ to \MD \ on regular graphs. 
\end{lemma}
Our reduction is based on an existing reduction from the maximum cut problem to the (uniform) sparsest cut problem developed by Bonsma et al.~\cite{Bonsma+12}; 
however, the difference from theirs is not totally trivial.
Before starting the proof of the lemma, we present a useful fact about the modularity density value for regular graphs: 
\begin{lemma}\label{lem:regular}
Let $G=(V,E)$ be a $d$-regular graph and $\mathcal{C}$ a partition of $V$. 
Then the modularity density can be rewritten as follows: 
\begin{align*}
D(\mathcal{C})= 4\sum_{C\in \mathcal{C}}\frac{|E(C)|}{|C|}-d|\mathcal{C}| \quad \text{and} \quad 
D(\mathcal{C})=d|\mathcal{C}| - 2\sum_{C\in \mathcal{C}}\frac{|E(C,V\setminus C)|}{|C|}. 
\end{align*}
In particular, if $|\mathcal{C}|=2$, say $\mathcal{C}=\{C,V\setminus C\}$, then the above second expression can be rewritten as 
\begin{align*}
D(\mathcal{C})=2d-2n\frac{|E(C,V\setminus C)|}{|C||V\setminus C|}. 
\end{align*}
\end{lemma}
\begin{proof}[Proof of Lemma~\ref{lem:regular}]
We have 
\begin{align*}
D(\mathcal{C})
=\sum_{C\in \mathcal{C}}\frac{2|E(C)|-|E(C,V\setminus C)|}{|C|}
=\sum_{C\in \mathcal{C}}\frac{2|E(C)|-(d|C|-2|E(C)|)}{|C|}
=4\sum_{C\in \mathcal{C}}\frac{|E(C)|}{|C|}-d|\mathcal{C}|
\end{align*}
and 
\begin{align*}
D(\mathcal{C})
=\sum_{C\in \mathcal{C}}\frac{2|E(C)|-|E(C,V\setminus C)|}{|C|}
=\sum_{C\in \mathcal{C}}\frac{d|C|-2|E(C,V\setminus C)|}{|C|}
=d|\mathcal{C}|-2\sum_{C\in \mathcal{C}}\frac{|E(C,V\setminus C)|}{|C|}. 
\end{align*}
The case of $|\mathcal{C}|=2$ is trivial. 
\end{proof}

\begin{proof}[Proof of Lemma~\ref{lem:reduction_to_MD}]
Let $(G,k)$ be an instance of \MC, where $G=(V,E)$ is a $3$-regular graph and $k$ is a positive integer. 
Let $n=|V|$. Note that $n\geq 4$. 
From this instance, we construct an instance $(\overline{G^*},r^*)$ of \MD \ on regular graphs as follows: 
For each $v\in V$, we construct two sets $I_v=\{v_0,v_1,\dots, v_{M-1}\}$ and $I'_v=\{v'_0,v'_1,\dots, v'_{M-1}\}$, 
each of which has size $M=n^3$. 
Note that for example, for a vertex named $u$, we have $I_u=\{u_0,u_1,\dots, u_{M-1}\}$ and $I'_u=\{u'_0,u'_1,\dots, u'_{M-1}\}$. 
For each $e=\{u,v\}\in E$, we create two edges $\{u_0,v_0\}$ and $\{u'_0,v'_0\}$. 
Hence, the subgraphs induced by $\{v_0\mid v\in V\}$ and $\{v'_0\mid v\in V\}$, respectively, are isomorphic to $G$. 
For each $v\in V$, we connect two subsets $I_v$ and $I'_v$ as follows: 
add an edge connecting $v_0$ and each of $\{v'_j\mid j=4,5,\dots, M-1\}$; 
add an edge connecting $v_i$ ($i=1,2,3$) and each of $\{v'_j\mid j=1,2,\dots, M-1\}$; 
add an edge connecting $v_i$ ($i=4,5,\dots, M-1$) and each of $\{v'_j\mid j=0,1,\dots, M-1,\ j\neq i\}$ (see Figure~\ref{fig:connection}). 
We denote by $G^*=(V^*,E^*)$ the graph constructed above. 
Note that $G^*$ is an $(M-1)$-regular graph defined on $|V^*|=2Mn$ vertices. 
Let $\overline{G^*}=(V^*,\overline{E^*})$ be a complement graph of $G^*$. 
In addition, set $r^* = 2M - 4 - \frac{12}{M} + \frac{8k}{Mn}$. 
As $\overline{G^*}$ is a $(2Mn-M)$-regular graph and $r^*$ is a positive real number, 
$(\overline{G^*}, r^*)$ is a valid instance of \MD. 
Obviously we can construct $(\overline{G^*},r^*)$ from $(G,r)$ in polynomial time. 
Thus, to prove the lemma, it suffices to show the following: 
$G$ has a cut $\{X,Y\}$ that satisfies $\textsf{val}(\{X,Y\})\geq k$ 
if and only if $\overline{G^*}$ has a partition $\mathcal{C}$ of $V^*$ with $|\mathcal{C}|\geq 2$ 
that satisfies $D(\mathcal{C})\geq r^*$. 
\begin{figure}
\centering
\includegraphics[scale=1.1]{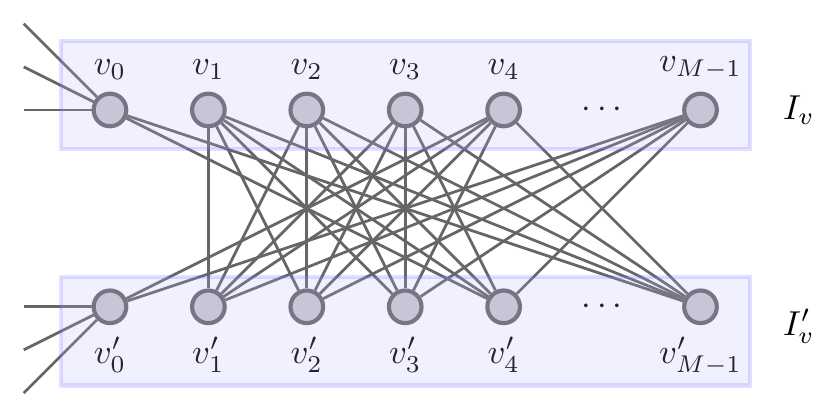}
\caption{The connection between $I_v$ and $I'_v$ (on $G^*$).}\label{fig:connection}
\end{figure}

Assume that $G$ has a cut $\{X,Y\}$ that satisfies $\textsf{val}(\{X,Y\})\geq k$. 
Now consider the partition $\mathcal{C}=\{C,V^*\setminus C\}$ of $V^*$ 
such that $C=\left(\bigcup_{v\in X} I_v\right) \cup \left(\bigcup_{v\in Y} I'_v\right)$. 
Then using Lemma~\ref{lem:regular}, we can evaluate the modularity density value of $\mathcal{C}$ as follows: 
\begin{align*}
D(\mathcal{C})
&=2(2Mn-M) - 4Mn\frac{|\overline{E^*}(C,V^*\setminus C)|}{(Mn)^2}\\
&=2(2Mn-M) - 4Mn\left(1-\frac{|E^*(C,V^*\setminus C)|}{(Mn)^2}\right)\\
&\geq 2(2Mn-M) - 4Mn\left(1-\frac{n(M^2-M-3)+2k}{(Mn)^2}\right)\\
&=2M - 4 - \frac{12}{M} + \frac{8k}{Mn}\\
&=r^*, 
\end{align*}
where the inequality follows from the fact that 
for each $v\in V$, two subsets $I_v$ and $I'_v$ are connected by $M^2-M-3$ edges on $G^*$ 
and the fact that $\textsf{val}(\{X,Y\})\geq k$. 
Therefore, we see that $\mathcal{C}$ is a partition that we desired. 

Assume that $\overline{G^*}$ has a partition $\mathcal{C}$ of $V^*$ with $|\mathcal{C}|\geq 2$ that satisfies $D(\mathcal{C})\geq r^*$. 
First we show that such a partition $\mathcal{C}$ has size exactly equal to two. 
Suppose for contradiction that $|\mathcal{C}|\geq 3$. 
Then using Lemma~\ref{lem:regular}, we have 
\begin{align*}
D(\mathcal{C}) 
&= 4\sum_{C\in \mathcal{C}}\frac{|\overline{E^*}(C)|}{|C|}-(2Mn-M)|\mathcal{C}|\\
&\leq 4\sum_{C\in \mathcal{C}}\frac{|C|-1}{2}-(2Mn-M)|\mathcal{C}|\\
&= 4Mn-2|\mathcal{C}|-(2Mn-M)|\mathcal{C}|\\
&\leq -2Mn+3M-6\\
&< r^*, 
\end{align*}
where the last inequality follows from $n\geq 4$. 
This contradicts to $D(\mathcal{C})\geq r^*$. 
Thus, we can represent $\mathcal{C}$ as $\mathcal{C}=\{C,V^*\setminus C\}$. 

Next we show that $|C|=|V^*\setminus C|=Mn$. 
Suppose for contradiction that $|C|\leq Mn-1$. 
Then using Lemma~\ref{lem:regular}, we have 
\begin{align*}
D(\mathcal{C})
&=2(2Mn-M) - 4Mn\frac{|\overline{E^*}(C,V^*\setminus C)|}{|C||V^*\setminus C|}\\
&=2(2Mn-M) - 4Mn\left(1-\frac{|E^*(C,V^*\setminus C)|}{|C||V^*\setminus C|}\right)\\
&\leq 2(2Mn-M) - 4Mn\left(1-\frac{(M-1)|C|}{|C|(2Mn-|C|)}\right)\\
&= -2M +\frac{4Mn(M-1)}{2Mn-|C|}\\
&\leq -2M +\frac{4Mn(M-1)}{Mn+1}\\
&=2M-4-\frac{12}{M}-\frac{4M-12n-4-\frac{12}{M}}{Mn+1}\\
&<2M-4-\frac{12}{M}+\frac{8k}{Mn}\\
&=r^*, 
\end{align*}
where the first inequality follows from the fact that $G^*$ is $(M-1)$-regular 
and the last inequality follows from $n\geq 4$. 
This contradicts to $D(\mathcal{C})\geq r^*$. 
Hence, we can represent $\mathcal{C}$ as $\mathcal{C}=\{C,V^*\setminus C\}$ in which $|C|=|V^*\setminus C|=Mn$ holds. 

Finally we show that for any $v\in V$, none of $I_v$ and $I'_v$ is divided by $\mathcal{C}$. 
Suppose for contradiction that for some $v\in V$, it holds that $I_v\cap C\neq \emptyset$ and $I_v\cap (V\setminus C)\neq \emptyset$. 
The following discussion also holds for $I'_v$. 
Recall that the number of edges between $I_v$ and $I'_v$ on $G^*$ is equal to $M^2-M-3$. 
An important fact is that among those $M^2-M-3$ edges, at most $M^2-2M+2$ edges are included in $E^*(C,V^*\setminus C)$, as $I_v$ is divided by $\mathcal{C}$. 
To see this, let $a=|I_v\cap C|$ and $b=|I'_v\cap C|$. 
Then we have 
\begin{align*}
&|E^*((I_v\cup I'_v)\cap C, (I_v\cup I'_v)\cap (V^*\setminus C))|\\
&=a(M-b)+b(M-a) - |\overline{E^*}((I_v\cup I'_v)\cap C, (I_v\cup I'_v)\cap (V^*\setminus C))|\\
&=(a+b)M -2ab - |\overline{E^*}((I_v\cup I'_v)\cap C, (I_v\cup I'_v)\cap (V^*\setminus C))|. 
\end{align*}
Consider the case where $a+b=M$ holds. Then as $I_v$ is divided by $\mathcal{C}$, we have $a\neq 0$ and $b\neq 0$, and so $ab\geq M-1$. 
Thus, we have 
\begin{align*}
|E^*((I_v\cup I'_v)\cap C, (I_v\cup I'_v)\cap (V^*\setminus C))|
\leq M^2-2(M-1)
=M^2-2M+2.
\end{align*}
Consider the other case where $a+b=\ell$ ($\ell=1,2,\dots, M-1$) holds. 
Then we have 
\begin{align*}
|E^*((I_v\cup I'_v)\cap C, (I_v\cup I'_v)\cap (V^*\setminus C))|
&=\ell M-2ab - |\overline{E^*}((I_v\cup I'_v)\cap C, (I_v\cup I'_v)\cap (V^*\setminus C))|\\
&\leq \begin{cases}
\ell M - 0 - \ell  &(\text{if } b=0)\\
\ell M - 2(\ell -1) - 0 &(\text{otherwise})
\end{cases}\\
&\leq \ell(M-1)+1\\
&\leq  M^2-2M+2. 
\end{align*}
Therefore, using the above fact, we can evaluate the modularity density value of $\mathcal{C}$ as follows: 
\begin{align*}
D(\mathcal{C})
&=2(2Mn-M) - 4Mn\left(1-\frac{|E^*(C,V^*\setminus C)|}{(Mn)^2}\right)\\
&\leq 2(2Mn-M) - 4Mn\left(1-\frac{3n+(n-1)(M^2-M-3)+M^2-2M+2}{(Mn)^2}\right)\\
&=2M-4-\frac{4}{n}+\frac{20}{Mn}\\
&<2M-4-\frac{12}{M}+\frac{8k}{Mn}\\
&=r^*, 
\end{align*}
where the second inequality follows from $n\geq 4$. 
This contradicts to $D(\mathcal{C})\geq r^*$. 
Moreover, by a similar argument (where $M^2-2M+2$ in the above calculation is replaced by $0$), 
we can also show that for any $v\in V$, it is impossible that both $I_v$ and $I'_v$ are in the same cluster in $\mathcal{C}$. 
Therefore, we see that $\mathcal{C}$ can be represented as $\mathcal{C}=\{C,V^*\setminus C\}$ such that 
\begin{align*}
C=\left(\bigcup_{v\in X}I_v\right)\cup \left(\bigcup_{v\in Y}I'_v\right)
\end{align*}
for some cut $\{X,Y\}$ of $G$. 
Then we can evaluate the modularity density value of $\mathcal{C}$ as follows: 
\begin{align*}
D(\mathcal{C})=2M - 4 - \frac{12}{M} + \frac{8\cdot \textsf{val}(\{X,Y\})}{Mn}. 
\end{align*}
As $D(\mathcal{C})\geq r^*$, we have $\textsf{val}(\{X,Y\})\geq k$. Therefore, $\{X,Y\}$ is a cut that we desired, and we are done. 
\end{proof}

As a by-product of our reduction, we can see that the uniform sparsest cut problem is NP-hard even on regular graphs. 
This is a stronger result of the NP-hardness of the uniform sparsest cut problem on general graphs, proved by Bonsma et al.~\cite{Bonsma+12}. 
In the uniform sparsest cut problem, given an undirected graph $G=(V,E)$, 
we are asked to find a cut $\{C,V\setminus C\}$ of $G$ that minimizes $\frac{|E(C,V\setminus C)|}{|C||V\setminus C|}$. 
On the gadget constructed above 
(i.e., the instance of (the decision version of) the modularity density maximization problem with at least two clusters 
constructed from an instance of (the decision version of) the maximum cut problem on $3$-regular graphs), 
the modularity density maximization problem with at least two clusters is essentially the same as the uniform sparsest cut problem. 
In fact, any optimal solution $\mathcal{C}$ satisfies $|\mathcal{C}|=2$ 
and hence by Lemma~\ref{lem:regular}, the objective is equivalent to that of the uniform sparsest cut problem. 
\begin{theorem}
The uniform sparsest cut problem is NP-hard even on regular graphs. 
\end{theorem}

\subsection{NP-hardness of the auxiliary problem}

We investigate the computational complexity of the AP appearing in column generation for the modularity density maximization problem. 
Specifically, we study the following problem: 
Given an undirected graph $G=(V,E)$ and $\bm{\lambda}=(\lambda_v)_{v\in V}\in \mathbb{R}^V$, 
we are asked to find a vertex subset $S\subseteq V$ that maximizes $\frac{2|E(S)|-|E(S,V\setminus S)|}{|S|}-\sum_{v\in S}\lambda_v$. 
When solving the AP in a column generation algorithm, $\bm{\lambda}$ is determined by an optimal solution $\bm{\lambda^*}$ to the RDP. 
Here, to investigate the computational complexity in general, we suppose that $\bm{\lambda}$ is also a part of the input. 
Our goal is to prove the following theorem: 
\begin{theorem}\label{thm:main_AP}
The AP (with $\bm{\lambda}$ input) is NP-hard even for $(n-4)$-regular graphs and $\bm{\lambda} > \bm{0}$. 
\end{theorem}

We again introduce the decision version of the problem: 
\begin{problem}[\AP]
Given an undirected graph $G=(V,E)$, $\bm{\lambda}=(\lambda_v)_{v\in V}\in \mathbb{R}^V$, and a real number $r$, 
does there exist a vertex subset $S\subseteq V$ that satisfies $\frac{2|E(S)|-|E(S,V\setminus S)|}{|S|}-\sum_{v\in S}\lambda_v\geq r$?
\end{problem}

To prove Theorem~\ref{thm:main_AP}, it suffices to show the NP-completeness of \AP \ for $(n-4)$-regular graphs and $\bm{\lambda}>\bm{0}$. 
For a given vertex subset $S\subseteq V$, $\bm{\lambda}\in \mathbb{R}^V$, and a real number $r$, 
we can check if $\frac{2|E(S)|-|E(S,V\setminus S)|}{|S|}-\sum_{v\in S}\lambda_v\geq r$ in polynomial time; 
thus, \AP \ is in the class NP. 

To show the NP-completeness, we construct a polynomial-time reduction from the decision version of the maximum clique problem on $(n-4)$-regular graphs, which is known to be NP-hard~\cite{Aggarwal+98,Garey_Johnson_79}. 
Let $G=(V,E)$ be an undirected graph. A vertex subset $S\subseteq V$ is called a clique if every pair of vertices in the subset has an edge in the graph. 
Then we can introduce the problem: 
\begin{problem}[\Clique]
Given an $(n-4)$-regular graph $G=(V,E)$ and a positive integer $k\in \mathbb{Z}_{>0}$, 
does there exist a clique $S\subseteq V$ of size $k$?
\end{problem}

The following lemma guarantees the existence of the desired reduction: 
\begin{lemma}\label{lem:reduction_to_AP}
There exists a polynomial-time reduction from \Clique \ to \AP \ for $(n-4)$-regular graphs and $\bm{\lambda}>\bm{0}$. 
\end{lemma}

Before starting the proof of the lemma, we borrow the following useful fact from extremal graph theory: 
\begin{lemma}[(A corollary of) Tur\'an's theorem (Theorem~7.1.1 in Diestel~\cite{Diestel05})]\label{lem:Turan}
Let $G=(V,E)$ be an undirected graph and $k$ a positive integer. 
If $G$ does not contain a clique of size $k$, then $|E|\leq \left(1-\frac{1}{k-1}\right)\frac{|V|^2}{2}$ holds. 
\end{lemma}

\begin{proof}[Proof of Lemma~\ref{lem:reduction_to_AP}]
Let $(G,k)$ be an instance of \Clique, where $G=(V,E)$ is an $(n-4)$-regular graph and $k$ is a positive integer. 
From this instance, we construct an instance $(G^*,\bm{\lambda^*},r^*)$ of \AP \ as follows: 
We just use the same graph, i.e., $G^*=G$, and set $\lambda^*_v=\frac{2(k-1)}{k}$ for every $v\in V$, and $r^*=-(n-4)$. 
To prove the theorem, it suffices to show that $G$ has a clique $S\subseteq V$ of size $k$ if and only if $G^*$ has a vertex subset $S\subseteq V$ that satisfies $\frac{2|E(S)|-|E(S,V\setminus S)|}{|S|}-\sum_{v\in S}\lambda^*_v \geq r^*$. 
Note that as $G$ is $(n-4)$-regular, we have $\frac{2|E(S)|-|E(S,V\setminus S)|}{|S|}=\frac{4|E(S)|}{|S|}-(n-4)$. 

Assume that $G$ has a clique $S\subseteq V$ of size $k$. 
Then we have 
\begin{align*}
\frac{4|E(S)|}{|S|}-(n-4)-\sum_{v\in S}\lambda^*_v=2(|S|-1)-(n-4)-\frac{2(k-1)}{k}|S|=r^*, 
\end{align*}
meaning that $S$ is a vertex subset that we desired. 

Assume that $G$ does not have a clique of size $k$. 
Then for any $S\subseteq V$, $G[S]$ does not contain a clique of size $k$. 
Thus, by Lemma~\ref{lem:Turan}, we have that for any $S\subseteq V$, 
\begin{align*}
|E(S)| \leq \left(1-\frac{1}{k-1}\right)\frac{|S|^2}{2}<\frac{(k-1)|S|^2}{2k}, 
\end{align*}
and hence
\begin{align*}
\frac{|E(S)|}{|S|}< \frac{(k-1)|S|}{2k}. 
\end{align*}
Therefore, we have 
\begin{align*}
\frac{4|E(S)|}{|S|}-(n-4)-\sum_{v\in S}\lambda^*_v< \frac{2(k-1)}{k}|S|-(n-4)-\frac{2(k-1)}{k}|S|=-(n-4), 
\end{align*}
which completes the proof. 
\end{proof}

%% file: 7_conclusion.tex
\section{Conclusion}\label{sec:conclusion}
In this study, we have investigated modularity density maximization from both algorithmic and computational complexity aspects.
Specifically, we first accelerated column generation for the modularity density maximization problem. 
We first pointed out that the AP, i.e., the auxiliary problem appearing in column generation, can be viewed as a dense subgraph discovery problem.
Then we employed a well-known strategy for dense subgraph discovery, called the greedy peeling, for approximately solving the AP. 
Moreover, we reformulate the AP to a sequence of $0$--$1$ linear programming problems,
enabling us to compute its optimal value more efficiently and to get more diverse columns.
Computational experiments using a variety of real-world networks demonstrate the effectiveness of our proposed algorithm.
In particular, our algorithm obtains an optimal solution within only two hours
to an instance with more than 100 vertices that cannot be solved by the state-of-the-art methods in 24 hours.
Finally, we showed the NP-hardness of a slight variant of the modularity density maximization problem,
where the output partition has to have two or more clusters, as well as showing the NP-hardness of the AP. 

There are several directions for future research. 
The most interesting one is to design a more scalable exact or column generation algorithm 
for the modularity density maximization problem. 
Although our proposed algorithm is scalable than the state-of-the-art methods, 
it is still difficult to handle large-scale graphs. 
Another direction is to apply the techniques similar to that used in the present paper to other important optimization problems. 
Our observation suggests that when designing a column generation algorithm 
for an optimization problem based on the set partitioning formulation, we can see the nature of dense subgraph discovery. 
Finally investigating the computational complexity of the modularity density maximization problem rather than its variants is a trivial future work. 
Is it possible to prove the NP-hardness of the modularity density maximization problem by constructing a more sophisticated polynomial-time reduction from some problem?